\documentclass{llncs}

\usepackage{times}

\usepackage{amsmath,amsfonts,graphpap,amscd,mathrsfs,graphicx,lscape}
\usepackage{epsfig,amssymb,amstext,xspace}
\usepackage{subfigure}
\usepackage{paralist}

\usepackage{color}              
\usepackage{epsfig}

\usepackage[suppress]{color-edits}
\addauthor{et}{green}
\addauthor{vs}{red}
\addauthor{bl}{blue}

\usepackage[ruled,vlined]{algorithm2e}
\renewcommand{\algorithmcfname}{ALGORITHM}
\SetAlFnt{\small}
\SetAlCapFnt{\small}
\SetAlCapNameFnt{\small}
\SetAlCapHSkip{0pt}
\IncMargin{-\parindent}

\newtheorem{defn}[theorem]{Definition}
\newtheorem{observation}[theorem]{Observation}

\def\blksquare{\rule{2mm}{2mm}}
\def\qedsymbol{\blksquare}

\def\blksquare{\rule{2mm}{2mm}}
\def\qedsymbol{\blksquare}
\newcommand{\bg}[1]{\medskip\noindent{\bf #1}}
\newcommand{\ed}{{\hfill\qedsymbol}\medskip}

\newenvironment{proofof}[1]{{\it{Proof of #1 : }}}{\ed}

\newcommand{\E}{\ensuremath{\textsc{E}}}
\newcommand{\R}{\ensuremath{\mathbb{R}}}

\newcommand{\opt}{\hbox{OPT}}

\newcommand{\Al}{\ensuremath{\mathcal{A}}}
\newcommand{\B}{\ensuremath{\mathcal{B}}}
\newcommand{\Aug}{\ensuremath{\textsc{Aug}}}

\newcommand{\tvec}{\ensuremath{{\bf t}}}
\newcommand{\avec}{\ensuremath{{\bf a}}}
\newcommand{\svec}{\ensuremath{{\bf s}}}

\newcommand{\F}{\ensuremath{\mathcal{F}}}
\newcommand{\mS}{\ensuremath{\mathcal{S}}}
\newcommand{\IG}{\ensuremath{\mathcal{IG}}}

\newcommand{\OurTitle}{Price of Stability in Games of Incomplete Information}

\begin{document}

\markboth{V. Syrgkanis}{\OurTitle}
\title{\OurTitle}
\author{Vasilis Syrgkanis
}

\institute{
Microsoft Research\\ {\tt vasy@microsoft.com}
}

\pagestyle{plain}

\maketitle

\begin{abstract}
We address the question of whether price of stability results (existence of equilibria with low social cost) are robust to incomplete information. We show that this is the case in potential games, if the underlying algorithmic social cost minimization problem admits a constant factor approximation algorithm via strict cost-sharing schemes. Roughly, if the existence of an $\alpha$-approximate equilibrium in the complete information setting was proven via the potential method, then there also exists a $\alpha\cdot \beta$-approximate Bayes-Nash equilibrium in the incomplete information setting, where $\beta$ is the approximation factor of the strict-cost sharing scheme algorithm. We apply our approach to Bayesian versions of the archetypal, in the price of stability analysis, network design models and show the existence of $O(\log(n))$-approximate Bayes-Nash equilibria in several games whose complete information counterparts have been well-studied, such as undirected network design games, multi-cast games and covering games.
\end{abstract}

\section{Introduction}

Quantifying the inefficiency caused by selfish user behavior in large engineered systems, such as communications networks, has been a central goal in the intersection of computer science and game theory. A large amount of literature, starting from the early works of \cite{Koutsoupias1999} and \cite{Roughgarden2002}, has analyzed the \emph{price of anarchy}: the ratio of the social cost of the worst selfish equilibrium outcome over the optimal social cost that a central planner could implement. At the counterpart of this approach is the notion of the \emph{price of stability} introduced in \cite{Anshelevich2004}, which is a more optimistic inefficiency quantification, looking at the best rather than worst equilibrium outcome. The intuition behind this notion is that such an outcome is the best individually stable outcome that a central planner can propose to selfish users. 

The \emph{price of stability} has been extensively analyzed in many network design scenarios, where selfish users interact over a network, each trying to optimize some individual objective \cite{Anshelevich2004,Bilo2010,Anshelevich2011,Bilo2013}. However, prior to this work, all the literature on quantifying the \emph{price of stability} has considered only complete information models, where the users know all the parameters of the game that they are playing, such as the location of all the users in the network and the objective of each user.

The goal of this work is to relax this informational assumption, and understand whether such optimistic efficiency results carry over to settings where players have incomplete information about the game. Such settings are predominantly modeled via Bayesian games \cite{Harsanyi}, where instead of complete information, players have distributional beliefs about the parameters of the game and are maximizing their objective in expectation over these beliefs. Such a unilaterally, stable in-expectation outcome is referred to as the Bayes-Nash equilibrium of the game. In this work we will analyze how the 
expected quality of the best such Bayes-Nash equilibrium (in expectation over the parameters of the game), compares
to the expected optimal solution that a central planner who is aware of all the parameters of the game, could have implemented. 
We will refer to the ratio of these two expected values as the \emph{Bayesian price of stability}.

The prototypical scenario, where the \emph{price of stability} analysis has been applied is that of network design games: a set of self-interested network users want to build a set of edges of a communication network to satisfy some connectivity constraint, such as connect a source to some root (multi-cast) or connect a source node to a target node. Each edge of the network has a building or usage cost, which is shared equally among all users that have chosen to use the edge. Such a model applies to large communication networks, where the built network  corresponds to 
a virtual overlay or a multicast tree on the Internet \cite{Herzog1995,Feigenbaum2001}. 

Incomplete information models are more suitable for such large scale communication networks, where the users are not aware of the location of all users in the network. Moreover, running a pre-processing protocol to learn the location and objectives of all users might be too costly and time-consuming, and even more importantly runs into mis-reporting problems. Motivated by such concerns, we analyze the incomplete information version of network design games, where the connectivity
objective of each user is private and drawn from some commonly known distribution. Such network design games of incomplete information where already studied by Alon et al. \cite{Alon2010}, however they did not analyze the question of whether price of stability guarantees persist in such settings and their techniques can only lead to Bayesian price of stability upper bounds that degrade with the size of the underlying graph, rather than the constant blow-ups we show here, with respect to the complete information price of stability. 

One of the engineering oriented motivations of analyzing the \emph{Bayesian price of stability} in such network design games of incomplete information is similar to their complete information counterparts: the designer of the network can propose a protocol that consists of a set of ``default" strategies for satisfying each connectivity objective over the network. For instance, the protocol could provide default paths from each potential source node to each potential target node or root. These paths are stable, or equivalently constitute a Bayes-Nash equilibrium, only if for each player, using the default path, minimizes his cost among all possible paths, and in expectation over the location of other users, assuming the rest of the users use the protocol. If the protocol doesn't satisfy this constraint then eventually self-interested individuals will start violating it, leading to potentially large inefficiencies. Thus the quality of the best Bayes-Nash equilibrium, corresponds to the best stable protocol that the network designer can devise.

In the complete information setting, \cite{Anshelevich2004} showed that although inefficiency of the worst equilibrium outcome 
of such network deisgn games, can grow linearly with the number of users, the social cost of the best equilibrium outcome is only an $O(\log(n))$ blow-up to the minimum social cost, where $n$ is the number of users. In other words, it was shown that the \emph{price of stability} is at most $O(\log(n))$. Our goal is to provide similar guarantees for the \emph{Bayesian price of stability} of these games.

\paragraph{Our Contributions.} Our main result is that the $O(\log(n))$ bound on the price of stability of most well-studied network design games, extends with only a constant degradation to the incomplete information versions of the game.
Specifically, we show that this holds for multi-cast network design where each player wants to build a path to a global root and for the undirected network design game where each player wants to build a path from some user-specific source to some sink, as well as for several related versions of covering games. For the case of multi-cast games the extension holds for the case where each player's source is distributed independently and asymmetrically over the set of nodes of the network, while for the undirected network design game and for the covering games it holds when the source-target pair of each user is distributed independently and identically over the set of node pairs of the network.

Apart from the specific results on network design games, our work gives a framework for proving \emph{Bayesian price of stability} bounds for any incomplete information game that admits a potential function: 
 i.e. a global function that tracks the difference in the cost of a user after any unilateral deviation. 
The main technique for providing bounds on the price of stability for complete information potential games is 
what is known as the \emph{potential method} \cite{Anshelevich2004}: Every equilibrium of a potential game is a local minimizer of the potential function. If the potential of the game is closely related to the social cost, then the equilibrium that corresponds to the global minimizer of the potential function must achieve a social cost that is not much more than the optimal social cost, i.e. minimizing the potential approximately minimizes the social cost. 

Our approach on bounding the \emph{Bayesian price of stability} is a black-box direct extension approach. Specifically, our goal is to claim that if the complete information version of the game has a low \emph{price of stability} via a \emph{potential method} proof, then the same bound extends with small degradation to the \emph{Bayesian price of stability}. The hope that such a black-box direct extension could be feasible is partially re-enforced by the fact that a similar ``extension theorem'' approach has proven very successful for the \emph{price of anarchy} analysis \cite{Roughgarden2012,Syrgkanis2012}, where it is shown that every \emph{price of anarchy} bound for a complete information game, proven via a canonical, \emph{smoothness} proof \cite{Roughgarden2009}, essentially extends, without any degradation to incomplete information versions of the game. 

We show that such a direct extension is sometimes also feasible for the analysis of the \emph{Bayesian price of stability}. Specifically, we show that when the price of stability of the complete information game is bounded by some factor $\alpha$, via the \emph{potential method}, then the \emph{Bayesian price of stability} is upper bound by a factor $\alpha\cdot \beta$, where
$\beta$ is an optimization, incentive-free, parameter of the game, that we refer to as the information gap. The information gap asks what is the loss in efficiency caused solely by lack of information and not by selfishness and was already analyzed by Alon et al. \cite{Alon2010}, mostly providing negative results.

At a high level, using similar reasoning as in the complete information game, every Bayes-Nash equilibrium of the Bayesian game will correspond to a local minimizer of some expected potential function. Hence, the expected social cost, when the players use strategies that minimize the expected potential should be close to the strategy profile that minimizes the expected social cost. However, a player's strategy in the Bayesian game is very constrained, as a player has information only about his private type. Hence, even if the players' goal is to minimize the social cost of the game, their information constraints disallow them to implement the ex-post social cost minimizing solution that corresponds to each instantiation of player types. Comparing the expected social cost achievable by strategies that use only private information, to the expected social cost of a central planner that has access to all the information, is exactly captured in the notion of the information gap.

The main technical challenge is then to show that the information gap is constant for the network design games that we are interested in. 
We show that there is a strong connection between the existence of approximation algorithms with
strict cost sharing schemes \cite{Gupta2003,Gupta2007} and the information gap. Specifically, we show that if the 
complete information social cost minimization problem that is associated with the Bayesian game, admits a constant-approximate
algorithm with a strict cost-sharing scheme, then the information gap of the Bayesian game, when player's types are distributed i.i.d. is also a constant. Moreover, if the strict cost-sharing scheme is also cross-monotonic \cite{Gupta2005,Pal2003} then 
the information gap bound extends to the non-i.i.d. setting. It is known that the Steiner Forest and the Vertex Cover problem associated with the undirected network design games and with the covering games admit a strict cost sharing scheme, while the Steiner Tree problem which is associated with the multi-cast game admits a strict cost sharing scheme that is also cross-monotone. Thus using our framework, these existing cost-sharing results imply a constant bound on the information gap and subsequently a constant degradation between the \emph{price of stability} and the \emph{Bayesian price of stability}.

\paragraph{Related Work.}
The \emph{price of stability} was introduced by \cite{Anshelevich2004} and has since been applied to many game theoretic models of networks \cite{Anshelevich2004,Bilo2010,Anshelevich2011,Bilo2013} and in other contexts, such as in discrete choice models of opinion \cite{Chierichetti2013}. However, as noted before, all the literature on the price of stability
assumes complete information models. A very interesting future direction is applying our approach to other models 
where the price of stability has been analyzed via the potential method, such as the discrete choice model of \cite{Chierichetti2013}, where we can analyze the case when the opinion of each player is private. 

The complete information counterparts of the network design games that we analyze here, have been extensively studied from several perspectives. \cite{Charikar2008,Chekuri2007} analyze the efficiency in multi-cast games where users arrive in an online manner and each user connects greedily to the root, when he arrives. The online setting has a similar spirit as our motivation in that users play in a manner that is unaware of opponent locations. However, the approach is very different and in \cite{Charikar2008,Chekuri2007}, it is assumed that players wont take into account any distributional knowledge about their opponents. \cite{Epstein2009} analyze the worst outcomes of the simultaneous network design game that are robust to coalitional deviations (strong Nash equilibria) and show that their quality matches the $O(\log(n))$ bound of the best Nash equilibria.  \cite{Balcan2010}, analyze best response dynamics in network design games, where the social planner proposes solutions that are followed by the players with some small probability. It is an interesting future direction to analyze versions of best-response dynamics that are natural under incomplete information. \cite{Leme2012} analyzes a special subclass of network design games and argues that when players arrive sequentially, the worst subgame-perfect equilibrium of the sequential game 
also achieves the $O(\log(n))$ bound of the best Nash equilibrium of the simultaneous game. 

Games of incomplete information have been recently extensively studied from the price of anarchy perspective, focusing mainly on models that capture auction scenarios: 
\cite{Christodoulou2008,Bhawalkar2011,Hassidim2011,Feldman2013,
Caragiannis2011,Syrgkanis2013}. Similar to our direct extension approach on the price of stability, \cite{Roughgarden2012} and \cite{Syrgkanis2012}, showed that a price of anarchy bound proven for the complete information setting via the smoothness framework proposed in \cite{Roughgarden2009}, directly extends with no degradation to the incomplete information version of the game. However, network design games are not smooth. 
Moreover, their price of anarchy grows linearly with the number of players as opposed to the price of stability. 

Strict cost sharing schemes, which is our main technical tool for proving upper bounds on the information gap, has played a quintessential role in designing simple polynomial time constant-factor approximation algorithms for social cost minimization problems associated with network design \cite{Gupta2003,Gupta2007}. This algorithmic framework 
was introduced in \cite{Gupta2003,Gupta2007} and subsequently used by many papers to design improved approximation algorithms for cost minimization games \cite{Fleischer2006}. Moreover, several papers have shown connection to stochastic versions of the algorithmic problem \cite{Gupta2004,Gupta2005,Gupta2011,Garg2008}. Our result can be viewed as an analogue to this extension from deterministic to stochastic versions of the problem, but for game-theoretic solution concepts and in fact our approach bears similarities to the approach of \cite{Garg2008}.
\vsdelete{In \cite{Gupta2004} it was shown that
the existence such approximation algorithms with strict cost sharing schemes, implies approximation algorithms for two-stage stochastic versions of network design problems, where the designer makes an investment in the first stage, when edge costs are cheaper, and then augments the solution in the second stage after the set of clients is revealed. This two stage process is similar in spirit with the ex-ante and ex-interim stage of the Bayesian game. \cite{Gupta2005,Gupta2011}, extend this connection to multi-stage stochastic versions of the problem and show that the cross-monotonicity of the strict cost-sharing scheme implies approximation algorithms for multi-stage stochastic problems too. It is interesting that cross-monotonicity also plays an important role in extending our results from the i.i.d. Bayesian setting to the non-i.i.d. Bayesian setting. A more elaborate understanding of this connection is an interesting future direction. Our analysis on the connection of the information gap with strict cost-sharing schemes is most closely related to the analysis in \cite{Garg2008}. \cite{Garg2008} analyzes an online stochastic version of a network design problem where at each iteration a client is drawn i.i.d. from some distribution. The algorithm then needs to make an online decision after observing each client. The difference with that setting is that the algorithm can base it's decision at iteration $t$, on the location of all clients that have arrived in previous iterations. This is for instance not allowed by a strategy of a player in the Bayesian game. However, we show that an adaptation of the techniques of \cite{Garg2008}, implies a bound on the information gap for i.i.d. settings. Moreover, our analysis for cross-monotonic schemes extends beyond the i.i.d. setting that was analyzed in \cite{Garg2008}. It is of interest, but outside the scope of this paper, to revisit and extend the results of \cite{Garg2008}, for stochastic online problem, but for non-identical client distributions at each iteration, based on our analysis of cross-monotonic strict cost sharing schemes.} 
Cost-sharing schemes have also been studied from a game theoretic perspective in mechanism design scenarios associated with network games \cite{Pal2003,Immorlica2008}. 


\section{Bayesian Games and Quality of Best Bayes-Nash}\label{sec:prelim}
Formally a Bayesian game of n players 
is defined as follows: Each player has a type $t_i$ drawn from some prior probability distribution $\pi_i$ on some type space $T_i$.
The distributions $\pi_i$ are independent and common knowledge and we denote with $\pi=\pi_1\times\ldots\times \pi_n$. 
Each player has a set of feasible actions $\F_i(t_i)$ that depend on his type and which are a subset of some universal action space $A_i$. We denote with $A=A_1\times\ldots\times A_n$. The cost of a player is a function
$C_i: A\times T_i\rightarrow \R$, mapping a feasible action profile and his type to a cost. The strategy of a player is a function $s_i: T_i\rightarrow A_i$, that takes as input his type $t_i$ and outputs a feasible action $a_i\in \F_i(t_i)$. We will denote the space of these Bayesian strategies with $\mS_i$. 

The standard solution concept in a Bayesian game is that of the Bayes-Nash equilibrium (BNE), which is a profile of strategies $\svec=(s_1,\ldots,s_n)$, such 
that each player is minimizing his cost in expectation over the types and actions of his opponents:
\begin{equation}
\forall a_i'\in \F_i(t_i): \E_{\tvec_{-i}}\left[C_i(\svec(\tvec)~;~t_i)\right] \leq \E_{\tvec_{-i}}\left[C_i(a_i',\svec_{-i}(\tvec_{-i})~;~t_i)\right] 
\end{equation}

We are interested in analyzing the expected social cost at equilibrium, which is defined as the sum of the players' expected costs:
\begin{equation}
\textstyle{K(\svec) = \E_{\tvec}\left[\sum_i C_i(\svec(\tvec)~;~t_i)\right]}
\end{equation}
We want to compare this social cost with the expected optimal social cost that would be achievable by a central coordinator who would observes the types of all players and picks the optimal action profile for each realized set of types. For a  type profile $\vec{t}$ let: 
\begin{equation}
\textstyle{\opt(\tvec)=\arg\min_{\avec\in \times_{i} A_i(t_i)}\sum_i C_i(\avec~;~t_i)}
\end{equation}
We are interested in showing that there exists some Bayes-Nash equilibrium where the expected social cost is at most constant times that of the expected optimal social cost. This is formally captured by the \emph{Bayesian Price of Stability} of the game, defined as:
\begin{equation}
BPoS = \min_{\svec \text{ is BNE}} \frac{K(\svec)}{\E_{t}[\opt(\tvec)]}
\end{equation}

\paragraph{Bayesian Potential Games.}\label{sec:potential}
We focus on the special case of a Bayesian game where the type $t_i$ of each player only affects and determines his action space and not directly his cost function. Thus 
there is a universal cost function $c_i:A\rightarrow \R^+$ that defines a cost given any action profile $\vec{a}\in A$ 
such that: 
$C_i(\avec~;~t_i)= c_i(\avec)$ if $\avec\in \F_i(t_i)$ and $\infty$ o.w..

We assume that the universal cost functions $c_i(a)$ admit a universal potential $\Phi:A\rightarrow \R$, i.e. for any action profile $\avec\in A$ and for any $a_i'\in A_i$: $c_i(\avec)-c_i(a_i',\avec_{-i}) = \Phi(\avec)-\Phi(a_i',\avec_{-i})$. 
Moreover, we denote with $C(\avec)= \sum_{i} c_i(\avec)$ the \emph{universal social cost} for any action profile. We will refer to such a game as a \emph{Bayesian potential game}. 

We first observe that if we view the Bayesian game as a normal form game where the strategy space of each player is $\mS_i$ and the cost function is $K_i(\svec) = \E_{\tvec}[c_i(\svec(\tvec))]$, then it is a potential game with potential $\Psi(\svec)=\E_{\tvec}[\Phi(\svec(\tvec))]$.
\begin{observation}\label{lem:bayesian-potential}
The normal form game defined by a Bayesian potential game is a potential game with potential $\Psi(\svec)=\E_{\tvec}[\Phi(\svec(\tvec))]$.
\end{observation}

 Moreover, a Nash equilibrium of this normal form game is a Bayes-Nash equilibrium of the Bayesian potential game. The latter then implies that the strategy profile $\svec$ that minimizes the potential function $\Psi(\svec)$ is a Bayes-Nash equilibrium:
\begin{observation}
The strategy profile $\svec^*=\arg\min_{\svec} \Psi(\svec)$ is a Bayes-Nash equilibrium of the Bayesian potential game.
\end{observation}

\paragraph{Bayesian Potential Method and the Information Gap.}\label{sec:method}
We will focus on Bayesian potential games, for which the \emph{potential method} of bounding the price of stability (see \cite{Anshelevich2004}) is applicable to their complete information versions. More formally, suppose that the universal potential function $\Phi(\avec)$ is $(\lambda,\mu)$-close to the universal social cost $C(\avec)$ in the sense that:
\begin{equation}
\forall \avec\in A: \lambda \cdot C(\avec)\leq \Phi(\avec)\leq \mu\cdot C(\avec)
\end{equation}
then it was shown by \cite{Anshelevich2004} that the best Nash equilibrium of the complete information potential game has social cost at most $\frac{\mu}{\lambda}$
of the minimum social cost. 
\vsdelete{The proof of this fact is rather neat and we present it here for completeness: An action profile $\avec$ of the complete information potential game is a Nash equilibrium if and only if it is a local minimum of the universal potential function $\Phi(\cdot)$. Now consider the Nash equilibrium $\avec^*$ that is the global minimum of the 
universal potential function. Then by the $(\lambda,\mu)$-closeness of the potential function, it follows that this equilibrium has social cost at most $\frac{\mu}{\lambda}$ times
the minimum social cost (denoting by $\opt$ the cost minimizing action profile):
\begin{equation}
C(\avec^*) \leq \frac{1}{\lambda} \Phi(\avec^*) \leq \frac{1}{\lambda} \Phi(\opt(t)) \leq \frac{\mu}{\lambda} C(\opt)
\end{equation}}

We provide an adaptation of this approach for Bayesian potential games. When analyzing Bayesian games, there is an inherrent extra complication that 
is not present in complete information games: playing the action profile that pointwise minimizes the social cost for each instantiation of the type profile
is not a feasible strategy of the Bayesian game. The reason is that for the players to employ such an action they need to know the type of their opponents, which is
private information. Hence, this lack of information does not allow the straightforward application of the \emph{potential method} to Bayesian games and requires us to analyze the following portion: 
%
\begin{defn}[Information Gap] Let $\tilde{\vec{s}}=\arg\min_{\vec{s}\in \mS_1\times\ldots\times \mS_n} K(\vec{s})$. Then the information gap of a Bayesian game is:
$\IG = \frac{K(\tilde{\vec{s}})}{\E_{\vec{t}}[\opt(\vec{t})]}$.
\end{defn}
The information gap is the ratio of the best social cost achievable by players via strategies that use only their private information (and each player's goal is to minimize social cost rather than personal cost), over the expected optimal cost of a centralized coordinator that computes the optimal action having access to the private type of every player. 

\begin{theorem}[Bayesian Potential Method]
Consider a Bayesian potential game, where the universal potential function $\Phi(\cdot)$ is $(\lambda,\mu)$-close to the universal social cost function $C(\cdot)$. Then the BPoS is at most $\frac{\mu}{\lambda}\cdot \IG$
\end{theorem}
\begin{proof}
Consider the potential minimizing Bayes-Nash equilibrium $\vec{s}^*$ and the social cost minimizing strategy profile $\tilde{\vec{s}}$. Then:
\begin{align*}
K(\vec{s}^*)=~&\E_{\vec{t}}\left[ C(\vec{s}^*(\vec{t}))\right]\leq \textstyle{\frac{\E_{\vec{t}}\left[ \Phi(\vec{s}^*(\vec{t}))\right]}{\lambda}
 \leq~\frac{\E_{\vec{t}}\left[\Phi(\tilde{\vec{s}}(\vec{t}))\right]}{\lambda}\leq \frac{\mu K(\tilde{\vec{s}})}{\lambda} \leq \frac{\mu \IG}{\lambda}  \E_{\vec{t}}\left[\opt(\vec{t})\right]}
\end{align*}
\qed\end{proof}

Bounding the information gap is the main technical challenge and the topic of the remaining sections. We will  focus on Bayesian versions of network design games \cite{Anshelevich2004}, which are Bayesian potential games where the universal potential is $(1,\log(n))$-close to the universal social cost. We  show that in such games the information gap is constant.

\section{Bayesian Network Design Games}\label{sec:network-design}
Bayesian potential games include intuitive Bayesian versions of \emph{network design games}. We first give the general definition of a \emph{Bayesian network design game} and then we give subclasses that will be of interest in the paper. 

\begin{defn}[Bayesian Network Design Game]\label{defn:network-design} A Bayesian network design game is characterized by a set of ground elements $E$ and a set of players $[n]$. The type space $T_i$ is common for all players and is denoted with $T$, i.e. $T_i=T$ for all $i\in [n]$. The universal action space $A_i$ is also common for all players and is simply the power set of $E$, i.e. $A_i=2^{E}$.  Conditional on her type, $t_i\in T$, the action space of player $i$ is a collection of sets of ground elements $\F_i(t_i)\subseteq 2^E$.
Each element $e\in E$, is associated with a non-negative cost 
$c_e$, which can be thought as the cost of using element $e$. Given an action profile $\avec$, we denote with $n_e(\avec)$, the congestion of element $e$, i.e. the number of players whose action $a_i$ contains $e$. If many players use the same element, then they equally share the building cost. The total cost of a player 
when she chooses a feasible action $a_i\in A_i(t_i)$, is the sum of his cost shares:
$c_i(\avec)=\sum_{e\in a_i} \frac{c_e}{n_e(\avec)}$.
\end{defn}

It is known that the complete information counterparts of these games, where the type of each player is fixed and common knowledge, admit Rosenthal's potential \cite{Rosenthal}, which 
would correspond to the universal potential of the Bayesian game under our terminology:
$\Phi(\avec) = \sum_{e\in E} \sum_{t=1}^{n_e(\avec)} \frac{c_e}{t}$. Therefore, a Bayesian network design game is a Bayesian potential game. Moreover, as was shown in \cite{Anshelevich2004}, Rosenthal's potential is $(1,\log(n))$-close to the universal social cost, $C(\avec)=\sum_{e\in E: n_e(\avec)\geq 1} c_e$. 

Network design games become more meaningful when the set of ground elements
$E$ corresponds to edges of a network, while the feasible action space of a player for a specific type $A(t_i)$, corresponds to a set of paths on the network connecting a set of nodes. The type of each player is the connectivity constraint that the paths need to satisfy. The following examples of network design games will be of main interest:

\begin{example}[Bayesian multicast game with private sources]
The set of ground elements correspond to the edges of a metric graph $G=(V,E)$, which has a designated root  node $r\in V$. Each player $i$ is associated with a source $s_i$, which is her private type and is drawn independently from a distribution $\pi_i$ over the set of terminals $V$. Each player wants to connect his source to the root $r$. Thus the feasible actions $\F_i(s_i)$ of a player of type $s_i$, is the set of paths from source $s_i$ to the root $r$. 
\end{example}

\begin{example}[Bayesian undirected network design with private source-sink pairs] The set of ground elements correspond to the edges of an undirected graph $G=(V,E)$. Each player $i$ is associated with a node pair $t_i=(s_i,r_i)$, of a source $s_i\in V$ and a sink $r_i\in V$, which are her private type and which are drawn independently from a distribution $\pi_i$ over the set of node pairs. Each player wants to connect her source to her sink. Thus the feasible actions $\F_i(t_i)$ 
is the set of paths from $s_i$ to $r_i$. 
\end{example}

\section{Bounding the Information Gap in Network Design Games}\label{sec:gap}
To bound the information gap of a \emph{Bayesian network design game}, we need to find a strategy profile $\vec{s}\in \mS_1\times\ldots\times\mS_n$ of the Bayesian game, such that the expected cost of the resulting strategy profile is close to the expected ex-post minimum social cost. To address this problem we first 
analyze what is the equivalent ex-post social cost minimization problem of the network design games that we defined in Section \ref{sec:network-design}.

From Definition \ref{defn:network-design}, 
the ex-post cost-minimization problem for each instance $\tvec$ of 
a type profile asks for the set of ground elements $S\subseteq E$ of minimum total building cost, such that $S$ is the union of feasible actions of the players for their type. We can modify all the network design games in Section \ref{sec:network-design} such that if a set $S$ of ground elements is a feasible action 
for a player, then every superset of $S$, is also a feasible strategy (e.g. if a path from $s$ to $r$ is a feasible strategy in the multicast game, then every superset of a path is also a feasible strategy). At equilibrium, no player will ever use non-minimal strategies and hence the two games are equivalent. Then the minimization problem asks for a set $S$, such that there exists
$a_i\in F_i(t_i)$ with $a_i\subseteq S$, which minimizes, $\sum_{e\in S} c_e$.

For instance, the ex-post social cost minimization problem of the Bayesian multicast game is the minimum cost Steiner Tree problem and of the undirected network design game it is the Steiner Forest. 

\paragraph{High level idea.} We upper bound the information gap by constructing a strategy profile of the Bayesian game as follows: prior to receiving their private types
the players draw a random sample of a type profile $\vec{w}\in T_1\times \ldots\times T_n$ from the prior distribution of types $\pi$. Then they build a solution for this random sample
(e.g. build a spanning tree for some random sample of sources). Then in the interim stage, after they receive their private type $t_i$, each of them independently and based only on her private type, augments the solution constructed in the ex-ante stage to make it feasible for his realized private type (e.g. routes to the nearest node of the steiner tree built in the first stage).  

The key property needed is that the solution of the first stage is constructed via some approximation algorithm which  admits \emph{cheap augmentation} when new players arrive. This is exactly captured by the existence of an algorithm that admits a strict cost sharing scheme, a well-studied notion which we provide here for completeness. 

\paragraph{Approximation algorithms with strict cost sharing schemes for sub-additive problems.}

Consider the following setting: Let $T$ be a universe of possible clients and $E$ a set of elements used to define a solution. For a set of clients $U\subseteq T$ let $Sol(U)\subseteq 2^E$ be the set of feasible solutions for $U$. The problem is \emph{sub-additive} if for any $U_1, U_2 \subseteq T$, if $E_1\in Sol(U_1)$ and $E_2\in Sol(U_2)$, then $E_1\cup E_2 \in Sol(U_1\cup U_2)$. The cost of a solution $F\subseteq E$ is $C(F)=\sum_{e\in F}c(e)$. Let $\opt(U)\in Sol(U)$ be the cost minimizing solution.

The connection to the cost-minimization problem arising from a network design game is obvious. The universe of possible clients corresponds to the possible set of types of the
players in the game and a set of clients corresponds to a set of players with a specific type profile, which we can also think as a set of types. A set of ground elements $F\subseteq E$ is 
feasible 
for a set of typed players, if for each player $i$, there exists an action $a_i \in \F_i(t_i)$ with $a_i\subseteq F$. Thus cost minimization is a sub-additive problem.

Let $\Al: 2^T\rightarrow 2^E$, be an algorithm (not necessarily polynomial time for our purposes) that given a set of clients $U$, it outputs an $\alpha$-approximately optimal feasible solution: $\Al(U)$. Moreover, let $\B(F,x)$ be an algorithm (not necessarily polynomial time) that given a set of elements $F\in Sol(U)$ that is a solution to a set of clients $U$ and a new client $x\in T$, returns a set of elements $\B(U,x)\subseteq E$ such that: $F\cup \B(U,x)\in Sol(U\cup \{x\})$.

\begin{defn}[$(\alpha,\beta)$-Strict Cost Sharing Scheme] An $\alpha$-approximation algorithm $\Al$ and an augmentation algorithm $\B$, admit a $\beta$-strict cost sharing scheme if there exists a cost-sharing function $\xi: 2^T\times T\rightarrow \R^+$, such that for any set of clients $U\subseteq T$ and for any client $x\in T$, $\xi(U,x)$ satisfies the following properties:
\begin{compactitem}
\item If $x\notin U$ then $\xi(U,x)=0$
\item Competitiveness: $\sum_{x\in U}\xi(U,x)\leq c(\opt(U))$
\item $\beta$-strictness with respect to $(\Al,\B)$: $c(\B(\Al(U),x))\leq \beta\cdot \xi(U\cup\{x\},x)$.
\end{compactitem}
\end{defn}

\begin{example}[Steiner Tree, $\alpha=1$ and $\beta=2$]\label{ex:steiner-tree}
As an example, consider the Steiner-Tree problem: an undirected metric graph $G=(V,E)$, a root $r$ and a set $D\subseteq V$ of clients. The problem asks for a Steiner-Tree on $D\cup\{r\}$. The problem admits algorithms $\Al,\B$ and a cost-sharing scheme that is $2$-strict with respect to $(\Al,\B)$. Algorithm $\Al$ is simply the optimal (exponential) Steiner Tree  algorithm. Consider the augmentation algorithm that given a set $D\subseteq V$ of clients and a new client $x$, it routes the client via the shortest path to the optimal steiner tree on $D\cup\{r\}$. Thus the cost of the algorithm is $c(\B(\Al(U),x))=d(\opt(D),x)$. 

Consider the cost shares $\xi(D,x)=\frac{1}{2}d(D-\{x\},x)$, for any $x\in D$ and $0$ otherwise. The $2$-strictness follows by:
$c(\B(\Al(U),x))=d(\opt(D),x)\leq d(D,x)= 2\cdot \xi(D\cup\{x\},x)$. 
Also, observe that $\sum_{x\in D} \xi(D,x)=\frac{1}{2}\sum_{x\in D}d(D-\{x\},x)\leq \frac{1}{2}\cdot MST(D),$ since each client in the minimum spanning tree on node set $D$ has a parent edge of cost at least the minimum distance to the remaining nodes. Since MST on $D$ is a two approximation to the optimal Steiner tree on $D$ we get that $\sum_{x\in D} \xi(D,x)\leq  c(\opt(D))$.
\end{example}


\paragraph{Strict cost sharing schemes and i.i.d. Bayesian network design games.}
We now argue how strict cost sharing schemes imply a bound on the information gap of the game. In this section we focus on the case where players' types are distributed identically and independently according to some distribution $\rho$ on the type space $T$, i.e. $\pi_i=\rho$.\footnote{Observe that the information gap depends on belief assumptions and not only on the structure of the underlying complete information game.}

\begin{theorem}\label{thm:iid}
If the ex-post cost-minimization problem associated with a Bayesian network design game admits an $\alpha$-approximation algorithm $\Al$ and an augmentation algorithm $\B$ with a cost-sharing scheme $\xi$ which is $\beta$-strict with respect to $(\Al,\B)$, then the information gap when types are drawn i.i.d. from some distribution $\rho$ is at most $\beta+\alpha$.
\end{theorem}
\begin{proof}

We construct a set of Bayesian strategies $\vec{s}=(s_1,\ldots,s_n)$ as follows: we draw a set $D=\{w_1,\ldots,w_n\}$ of random samples of $n-1$ clients/types according to $\pi=\rho^n$. Then we compute a feasible solution $\Al(D)$ on this random draw of clients using algorithm $\Al$. Let $R=\{t_1,\ldots,t_n\}$ denote the set of true types/clients of the players. At the interim stage (after viewing the true type), each player computes the augmentation required to serve his true type, on top of $\Al(D)$, using the augmentation algorithm $\B$, i.e. his strategy is $s(t_i)=\Al(D)\cup \B(\Al(D),t_i)$. In fact, we can define $s_i(t_i)$ as the cheapest subset of $\Al(D)\cup \B(\Al(D),t_i)$ that remains a feasible action. Observe that each player uses only information available to him at the interim stage.


It remains to bound the expected cost of the final solution implied by the strategies constructed by the above process (the process can also be derandomized). For any instance of the real types $R$ and fake types $D$, the cost of the algorithm is the cost of the set: $\Al(D)\cup (\cup_{i\in [n]} \B(\Al(D),t_i))$. The expected cost of $\Al(D)$ is bounded by: 
\begin{align*}
\E_{D\sim \rho^{n-1}}\left[c(\Al(D))\right]\leq \alpha\cdot \E_{D\sim \rho^{n-1}}\left[c(\opt(D))\right]
\leq\alpha \cdot \E_{R\sim \rho^n}\left[c(\opt(R))\right],
\end{align*}
since optimal cost is monotone in the client set. It remains to bound the cost of the augmentations, denoted by $\Aug$, by $\beta\cdot \E_{R\sim \pi^n}\left[c(\opt(R))\right]$.
By symmetry:
\begin{align*}
\textstyle{\Aug = n \cdot \E_{t\sim \rho, D\sim \rho^{n-1}}\left[c(\B(\Al(D),t))\right]\leq n\cdot \beta\cdot \E_{t\sim \rho, D\sim \rho^{n-1}}\left[\xi(D\cup\{t\},t)\right]}
\end{align*}
Observe that the latter expectation is equivalent to the following portion: draw a random sample of $n$ clients according to $\rho^n$ and then pick one client at random and account for his cost-share when adding him to the rest of the client set.\footnote{This regrouping trick was also used in \cite{Garg2008}.} From this we get:
\begin{align*}
\Aug \leq~& \textstyle{n\cdot \beta\cdot \E_{t\sim \rho, D\sim \rho^{n-1}}\left[\xi(D\cup\{t\},t)\right]= n\cdot \beta\cdot  \E_{R\sim \rho^n}\left[\frac{1}{n}\sum_{t\in R}\xi(R,t)\right]}\\
=~&\textstyle{\beta\cdot \E_{R\sim \rho^n}\left[\sum_{t\in R}\xi(R,t)\right]\leq ~\beta\cdot \E_{R\sim \rho^n}\left[c(\opt(R))\right]},
\end{align*}
where the last inequality follows from competitiveness of the strict cost shares. 
\qed\end{proof}

Theorem \ref{thm:iid} implies several results on the information gap for Bayesian games associated with many optimization problems. For instance, we saw in the previous section that the Steiner Tree problem has a $1$-approximation algorithm with a $2$-strict cost sharing scheme. Moreover, using the existence of strict cost-sharing schemes in the literature \cite{Gupta2004,Gupta2007,Garg2008}, we get the following bounds and the main theorem of this section:\vsdelete{. that the Steiner Forest has a $2$-approximation algorithm with a $3$-strict cost sharing scheme and the vertex cover with has
a $3$-approximation algorithm with $3$-strict cost shares ($d+1$, for $d$-Hypergraph cover instead of $3$). All these problems have corresponding Bayesian Network Design Games. Thus we get the following corollaries:}
\begin{theorem}
When types are i.i.d., the information gap for:
i) Bayesian multicast game is at most $3$, ii) Bayesian undirected network design game is at most $5$, iii) Bayesian vertex cover game (App. \ref{sec:vertex-cover}) is at most $6$ (and $2(d+1)$ for $d$-hypergraph cover). Hence, for all these games the BPoS is $O(\log(n))$ when types are i.i.d..
\end{theorem}

\paragraph{Cross-monotonicity and non-i.i.d. Bayesian network design games.}
We also extend the previous analysis to non-i.i.d. Bayesian games, if we impose the extra restriction on the optimization problem, that the cost-shares are also cross-monotone, i.e. $\xi(D,x)$ is monotone non-increasing in $D$. 

\begin{theorem}\label{thm:cross-monotone}
If the ex-post cost-minimization problem associated with a Bayesian network design game admits an $\alpha$-approximation algorithm $\Al$ and an augmentation algorithm $\B$ with a cost-sharing scheme $\xi$ which is $\beta$-strict with respect to $(\Al,\B)$ and cross-monotonic, then the information gap of the corresponding Bayesian game when each player's type is drawn independently from some distribution $\pi_i$, is at most $\beta+\alpha$.
\end{theorem}
\begin{proof}
We slightly alter the process in the proof of Theorem \ref{thm:iid} to take into account that clients are not drawn identically. 
\begin{algorithm}\label{alg:noniid}
Let $R=(t_1,\ldots,t_n)$ be the input clients, which are drawn independently and identically according to $\pi=(\pi_1,\ldots,\pi_n)$\\
Draw a random sample of $n$ clients $D=(d_1,\ldots,d_n)\sim \pi$\\
Compute $\Al(D)$\\
Output: $\forall i\in [n]: s_i(t_i) = \Al(D)\cup \B(\Al(D),t_i)$
\caption{Preprocessing, random-sampling algorithm for constructing strategies of Bayesian game with non-identical distributions.}
\end{algorithm}
Intuitively, the algorithm draws a random sample of $n$ clients according to the commonly known distribution of types. Then it computes a feasible solution on this random draw of clients using algorithm $\Al$.  Subsequently, each player computes the augmentation required to serve his true type, on top of the randomly sampled ones in the first stage. Then his strategy is $s(t_i)=\Al(D)\cup \B(\Al(D),t_i)$, which we know that it is in $Sol(\{t_i\})$ by the definition of the augmentation algorithm. 

Similar to Theorem \ref{thm:iid} it suffices to bound the expected cost of the final solution implied by the strategies constructed by Algorithm \ref{alg:noniid}, which is the expected cost of the set of elements: $\Al(D)\cup (\cup_{i\in [n]} \B(\Al(D),t_i))$. 

First, observe that the expected cost of the solution $\Al(D)$ on the randomly sampled clients $D$ is at most: 
$$\E_{D\sim \pi}\left[c(\Al(D))\right]\leq \alpha\cdot \E_{D\sim \pi}\left[c(\opt(D))\right]=\alpha \cdot \E_{R\sim \pi}\left[c(\opt(R))\right],$$ 
Thus it remains to bound the cost of the augmentations by at most $\beta\cdot \E_{R\sim \pi}\left[c(\opt(R))\right]$.
The expected cost of augmenting input client $i$ is:
\begin{align*}
\Aug_i =& \E_{t_i\sim \pi_i, D\sim \pi}\left[c(\B(\Al(D),t_i))\right]\leq \beta\cdot \E_{t_i\sim \pi_i, D\sim \pi}\left[\xi(D\cup\{t_i\},t_i)\right]
\end{align*}
By the cross-monotonicity of the cost shares:
\begin{align*}
\Aug_i \leq &  \beta\cdot \E_{t_i\sim \pi_i, D\sim \pi}\left[\xi(D_{-i}\cup\{t_i\},t_i)\right]
=  \beta\cdot \E_{t_i\sim \pi_i, D_{-i}\sim \pi_{-i}}\left[\xi(D_{-i} \cup\{t_i\},t_i)\right]\\
=&  \beta\cdot \E_{t_i\sim \pi_i, R_{-i}\sim \pi_{-i}}\left[\xi(R_{-i} \cup\{t_i\},t_i)\right] = 
\beta\cdot \E_{R\sim \pi}\left[\xi(R,t_i)\right]
\end{align*}
Using the competitiveness of the cost shares we can bound the total cost of the augmentation by the desired quantity:
\begin{align*}
\sum_i \Aug_i \leq \beta\cdot\sum_i \E_{R\sim \pi}\left[\xi(R,x)\right]= \beta\cdot  \E_{R\sim \pi}\left[\sum_{i}\xi(R,x_i)\right] \leq \beta\cdot \E_{R\sim \pi}\left[c(\opt(R))\right]
\end{align*}
Thus we have established that the expected cost of the set of elements: $\Al(D)\cup (\cup_{i\in [n]} \B(\Al(D),t_i))$, is at most $(\beta+\alpha)\cdot  \E_{R\sim \pi}\left[c(\opt(R))\right]$.
\end{proof}

By cross monotonicity of the strict cost-shares for the Steiner Tree problem, we get:
\begin{theorem}
When distributions of types are independent, the information gap for the Bayesian multicast game is at most $3$ and the Bayesian price of stability is $O(\log(n))$.
\end{theorem}

An interesting case of the non-i.d.d. Bayesian multicast game is the \emph{independent decisions model} \cite{Gupta2004}: each node in the graph $G$ is associated with a player $i$ and this player enters the game with some probability $p_i$, or is not present with some probability $1-p_i$. One can view this setting as a Bayesian multicast game among $|V|$ players, where the distribution $\pi_i$ of a player associated with a node $i$ is a two point distribution that puts a probability $p$ on the node associated with player $i$ and probability $1-p$ on the root $r$. 

\vsdelete{The existence of cross-monotonic strict cost sharing schemes for the other problems that we are interested in, is a very interesting long-standing open question in the literature (see \cite{Gupta2011}), since it has implications on multi-stage stochastic versions of the optimization problems.} 

\vscomment{
We should look at the Sigal's discrete choice models and show whether the PoS of 2 carries over to the setting where the opinion of each player is private information and drawn from some commonly known distribution.}

\vscomment{We should also refer to the simultaneous second price auction, PoS=1 for complete information.}

\bibliographystyle{plain}
\bibliography{price-of-stability}

\begin{appendix}

\section{Vertex Cover Games}\label{sec:vertex-cover}
 Consider the following setting: there is a set of nodes $v\in V$, each associated with a cost $c_v$. Each player $i\in [n]$, is associated with an edge $(u,v)\in V\times V$, which represents the pair of nodes that
can satisfy (coiver) player $i$'s needs. This edge is the private information of player $i$ and is drawn from some distribution $\pi_i$ on the set of node pairs. The action of a player 
$i$ with type $t_i=(u,v)\in V\times V$, is to choose one of the two nodes $u$ and $v$. If many players pick the same node then they share equally the cost of the node. This game can
also be cast as a Bayesian Network Design Game where the ground set of elements $E$ is equal to the set of nodes $V$ and the action set of a player with type $t_i=(u,v)$
is the collection of singletons $\{\{u\},\{v\}\}$.
 
\paragraph{$d$-Hypergraph Cover} More generally, we can also consider the generalization of the above game where each player is not associated with an edge $(u,v)$, but rather with a hyperedge $h\subseteq V$ of size $d$, which is his private type and is drawn from some distribution $\pi_i$ on the set of hyperedges on $V$. Then the set of feasible actions of a player with type $h$, is to choose one node in his hyperedge. Intuitively, the nodes in the hyperedge are the nodes that can "satisfy" player $i$. By similar analogy, the above setting is also a Bayesian Network Design Game.

\section{Omitted Proofs}\label{sec:omitted}

\begin{proofof}{Observation \ref{lem:bayesian-potential}}
Consider a strategy profile $\svec$ of the Bayesian Potential Game and a strategy $s_i'$ for player $i$. Then:
\begin{align*}
K_i(\svec)-K_i(s_i',\svec_{-i}) =~& \E_{\tvec}\left[C_i(\svec(\tvec))-C_i(s_i'(t_i),\svec_{-i}(\tvec_{-i}))\right] \\
=~& \E_{\tvec}\left[c_i(\svec(\tvec))-c_i(s_i'(t_i),\svec_{-i}(\tvec_{-i}))\right]\\
=~& \E_{\tvec}\left[\Phi(\svec(\tvec))-\Phi(s_i'(t_i),\svec_{-i}(\tvec_{-i}))\right] = \Psi(\svec)-\Psi(s_i',\svec_{-i})
\end{align*}
\end{proofof}

\end{appendix}

\end{document}